\documentclass[%
reprint,
 amsmath,amssymb,
 aps,
pra,
showkeys,
]{revtex4-2}
\usepackage{titlesec}

\usepackage{mathrsfs}

\usepackage{graphicx}
\usepackage{dcolumn}
\usepackage{bm}

\usepackage{amsthm}
\usepackage{xspace}
\usepackage[mathcal]{euscript}

\usepackage{orcidlink}

\usepackage{url}
\usepackage{hyperref}
\usepackage{color}
\definecolor{refcolor}{RGB}{0,0,190}
\hypersetup{
    colorlinks,
    citecolor=refcolor,
    filecolor=refcolor,
    linkcolor=refcolor,
    urlcolor=refcolor
}

\usepackage{booktabs}

\usepackage{bookmark}
\bookmarksetup{
  numbered, 
  open,
}

\begin{document}

\newtheorem{theorem}{Theorem}
\newtheorem{question}{Question}

\theoremstyle{remark}
\newtheorem{answer}{Answer}

\newtheorem{proofStep}{Step}
\numberwithin{proofStep}{theorem} 

\theoremstyle{definition}

\newtheorem{observation}{Observation}
\newtheorem{option}{Option}
\newtheorem{problem}{Problem}
\newtheorem{consequence}{Consequence}
\newtheorem{experiment}{Experiment}
\newtheorem{implication}{Implication}
\newtheorem{thoughtExp}{Thought Experiment}
\newtheorem{claim}{Claim}

\newtheorem{proposition}{Proposition}
\newtheorem{principle}{Principle}
\newtheorem{lemma}{Lemma}
\newtheorem{corollary}{Corollary}
\newtheorem{definition}{Definition}
\newtheorem{hypothesis}{Hypothesis}
\newtheorem{assumption}{Assumption}
\newtheorem{property}{Property}
\newtheorem{criterion}{Criterion}
\newtheorem{objection}{Objection}
\newtheorem{remark}{Remark}
\newtheorem{example}{Example}
\theoremstyle{remark}
\newtheorem{reply}{Reply}

\newenvironment{solution}[2] {\paragraph{Solution to {#1} {#2} :}}{\hfill$\square$}



\newcommand{\orcid}[1]{\href{https://orcid.org/#1}{\textcolor[HTML]{A6CE39}{\aiOrcid}}}

\newcommand{\pbref}[1]{\ref{#1} (\nameref*{#1})}
   
\def\({\big(}
\def\){\big)}

\newcommand{\tn}{\textnormal}
\newcommand{\ds}{\displaystyle}
\newcommand{\dsfrac}[2]{\displaystyle{\frac{#1}{#2}}}

\newcommand{\boplus}{\textstyle{\bigoplus}}
\newcommand{\botimes}{\textstyle{\bigotimes}}
\newcommand{\bcup}{\textstyle{\bigcup}}
\newcommand{\bcap}{\textstyle{\bigcap}}
\newcommand{\bsqcup}{\textstyle{\bigsqcup}}
\newcommand{\bsqcap}{\textstyle{\bigsqcap}}

\newcommand{\MQS}{\ref{def:MQS}}
\newcommand{\TPS}{\ref{as:TPS}}
\newcommand{\ASSpace}{\ref{as:space}}
\newcommand{\ASBasis}{\ref{as:basis}}

\newcommand{\struct}{\mc{S}}
\newcommand{\kind}{\mc{K}}

\newcommand{\statespace}{\mathcal{S}}
\newcommand{\hilbert}{\mathcal{H}}
\newcommand{\vectorspace}{\mathcal{V}}
\newcommand{\mc}[1]{\mathcal{#1}}
\newcommand{\ms}[1]{\mathscr{#1}}

\newcommand{\wh}[1]{\widehat{#1}}
\newcommand{\dwh}[1]{\wh{\rule{0ex}{1.3ex}\smash{\wh{\hfill{#1}\,}}}}

\newcommand{\wt}[1]{\widetilde{#1}}
\newcommand{\wht}[1]{\widehat{\widetilde{#1}}}
\newcommand{\on}[1]{\operatorname{#1}}

\newcommand{\vect}[1]{\mathsf{#1}}
\newcommand{\oper}[1]{\wh{\mathbf{#1}}}
\newcommand{\TT}{\intercal}

\newcommand{\R}{\mathbb{R}}
\newcommand{\C}{\mathbb{C}}
\newcommand{\Z}{\mathbb{Z}}
\newcommand{\K}{\mathbb{K}}
\newcommand{\N}{\mathbb{N}}
\newcommand{\T}{\mathbb{T}}
\newcommand{\Prj}{\mathcal{P}}
\newcommand{\abs}[1]{\left|#1\right|}

\newcommand{\de}{\operatorname{d}}
\newcommand{\tr}{\operatorname{tr}}
\newcommand{\im}{\operatorname{Im}}

\newcommand{\dof}{d.o.f.\xspace}
\newcommand{\dofs}{d.o.f.s\xspace}

\newcommand{\ie}{\textit{i.e.}\ }
\newcommand{\vs}{\textit{vs.}\ }
\newcommand{\eg}{\textit{e.g.}\ }
\newcommand{\cf}{\textit{cf.}\ }
\newcommand{\etc}{\textit{etc}}
\newcommand{\etal}{\textit{et al.}}

\newcommand{\Span}{\tn{span}}
\newcommand{\pde}{PDE}
\newcommand{\U}{\tn{U}}
\newcommand{\SU}{\tn{SU}}
\newcommand{\GL}{\tn{GL}}

\newcommand{\schrod}{Schr\"odinger}
\newcommand{\vonneum}{Liouville-von Neumann}
\newcommand{\ks}{Kochen-Specker}
\newcommand{\leggarg}{Leggett-Garg}
\newcommand{\bra}[1]{\langle#1|}
\newcommand{\ket}[1]{|#1\rangle}
\newcommand{\braket}[2]{\langle#1|#2\rangle}
\newcommand{\ketbra}[2]{|#1\rangle\langle#2|}
\newcommand{\expectation}[1]{\langle#1\rangle}
\newcommand{\Herm}{\tn{Herm}}
\newcommand{\Eval}{\tn{Eval}}
\newcommand{\Sym}[1]{\tn{Sym}_{#1}}
\newcommand{\meanvalue}[2]{\langle{#1}\rangle_{#2}}

\newcommand{\btimes}{\boxtimes}
\newcommand{\btimess}{{\boxtimes_s}}

\newcommand{\h}{\mathbf{(2\pi\hbar)}}
\newcommand{\x}{\mathbf{x}}
\newcommand{\xThree}{\boldsymbol{x}}
\newcommand{\z}{\mathbf{z}}
\newcommand{\q}{\boldsymbol{q}}
\newcommand{\p}{\boldsymbol{p}}
\newcommand{\f}{\mathbf{f}}
\newcommand{\0}{\mathbf{0}}
\newcommand{\annih}{\widehat{\mathbf{a}}}

\newcommand{\cs}{\mathscr{C}}
\newcommand{\ps}{\mathscr{P}}
\newcommand{\xhat}{\widehat{\x}}
\newcommand{\phat}{\widehat{\mathbf{p}}}
\newcommand{\fqproj}[1]{\Pi_{#1}}
\newcommand{\cqproj}[1]{\wh{\Pi}_{#1}}
\newcommand{\cproj}[1]{\wh{\Pi}^{\perp}_{#1}}

\newcommand{\M}{\mathbb{E}_3}
\newcommand{\D}{\mathbf{D}}
\newcommand{\dn}{\tn{d}}
\newcommand{\db}{\mathbf{d}}
\newcommand{\n}{\mathbf{n}}
\newcommand{\m}{\mathbf{m}}
\newcommand{\V}[1]{\mathbb{V}_{#1}}
\newcommand{\F}[1]{\mathcal{F}_{#1}}
\newcommand{\Fvacuumfield}{\widetilde{\mathcal{F}}^0}
\newcommand{\nD}[1]{|{#1}|}
\newcommand{\Lin}{\mathcal{L}}
\newcommand{\End}{\tn{End}}
\newcommand{\vbundle}[4]{{#1}\to {#2} \stackrel{\pi_{#3}}{\to} {#4}}
\newcommand{\vbundlex}[1]{\vbundle{V_{#1}}{E_{#1}}{#1}{M_{#1}}}
\newcommand{\rep}{\rho_{\scriptscriptstyle\btimes}}

\newcommand{\intl}[1]{\int\limits_{#1}}

\newcommand{\moyalBracket}[1]{\{\mskip-5mu\{#1\}\mskip-5mu\}}

\newcommand{\Hint}{H_{\tn{int}}}

\newcommand{\quot}[1]{``#1''}

\def\sref #1{\S\ref{#1}}

\newcommand{\dBB}{de Broglie--Bohm}
\newcommand{\dBBt}{{\dBB} theory}
\newcommand{\pwt}{pilot-wave theory}
\newcommand{\PWT}{PWT}
\newcommand{\NRQM}{{\textbf{NRQM}}}

\newcommand{\image}[3]{
\begin{figure}[!ht]
\centering
\includegraphics[width=#2\textwidth]{#1}
\caption{\small{\label{#1}#3}}
\end{figure}
}

\newcommand{\no}[1]{\overline{#1}}

\newcommand{\cmark}{\ding{51}}%
\newcommand{\xmark}{\ding{55}}%

\newcommand{\isonomic}{\hyperref[def:isonomy]{isonomic}\xspace}
\newcommand{\isonomy}{\hyperref[def:isonomy]{isonomy}\xspace}
\newcommand{\isonomies}{\hyperref[def:isonomy]{isonomies}\xspace}

\hyphenation{pa-ram-e-trized ob-serv-er }

\newcommand{\starbreak}{
\vspace{-0.1cm}
\begin{center}
$\ast\ast\ast$
\end{center}
\vspace{-0.2cm}
}
\newcommand{\todo}[1]{\textcolor{red}{\raisebox{-1pt}{\scalebox{1.5}{$\bigtriangleup$}}$\mkern-15.3mu${\textbf{!}}} \textcolor{blue}{\ #1}\PackageWarning{TODO:}{#1!}{}}



\title{Does the Hamiltonian determine the tensor product structure and the 3d space?}

\author{Ovidiu Cristinel Stoica\ \orcidlink{0000-0002-2765-1562}}
\affiliation{
 Dept. of Theoretical Physics, NIPNE---HH, Bucharest, Romania. \\
	Email: \href{mailto:cristi.stoica@theory.nipne.ro}{cristi.stoica@theory.nipne.ro},  \href{mailto:holotronix@gmail.com}{holotronix@gmail.com}
	}%

\date{\today} 

\begin{abstract}
It was proposed that the tensor product structure of the Hilbert space is uniquely determined by the Hamiltonian's spectrum, for most finite-dimensional cases satisfying certain conditions. 

I show that any such method would lead to infinitely many tensor product structures. The dimension of the space of solutions grows exponentially with the number of qudits. In addition, even if the result were unique, such a Hamiltonian would not entangle subsystems.

These results affect the proposals to recover the 3d space from the Hamiltonian.
\end{abstract}


\maketitle

\section{Introduction}
\label{s:intro}

The usual ingredients for specifying a quantum theory are a Hilbert space $\hilbert$, a time-dependent density operator $\oper{\rho}(t)$ on $\hilbert$ representing the total state of the system, a Hermitian operator $\oper{H}$ on $\hilbert$ -- the Hamiltonian operator -- encoding the time evolution of the system, a decomposition of the Hilbert space as a tensor product
\begin{equation}
\label{eq:hilbert-tensor}
\hilbert=\hilbert_1\otimes\hilbert_2\otimes\ldots,
\end{equation}
describing the structure of the system in terms of subsystems or regions of space, and a set of Hermitian operators on $\hilbert$ or on the subsystem spaces $\hilbert_j$, to represent physically observable properties like positions and momenta.

The decomposition \eqref{eq:hilbert-tensor} represents the \emph{tensor product structure} (TPS) of $\hilbert$.
If the system consists of subsystems, for example particles whose number remains constant, the Hilbert spaces $\hilbert_j$ correspond to these subsystems. Decompositions like \eqref{eq:hilbert-tensor} also appear when we are interested in the quantum fields on regions of space.

In practice, physicists inferred the TPS from experiments, that is, from the observables corresponding to subsystems. The TPS can be determined uniquely by a set of algebras of observables \cite{ZanardiLidarLloyd2004QuantumTensorProductStructuresAreObservableInduced}.

In \cite{CotlerEtAl2019LocalityFromSpectrum} it was proposed that under certain conditions there is a method determining, only from the Hamiltonian's spectrum (including the multiplicities of eigenvalues), a TPS supposed to be unique in most cases with $\dim \hilbert<\infty$.
All this without using other structures, observables, or a preferred basis.
This suggestion was used in the literature to propose that the $3$d space or spacetime emerges from the Hamiltonian \cite{CarrollSingh2019MadDogEverettianism,Carroll2021RealityAsAVectorInHilbertSpace,Giddings2019QuantumFirstGravity,Szangolies2023StandardModelSymmetryAndQubitEntanglement}.
In \cite{Stoica2021SpaceThePreferredBasisCannotUniquelyEmergeFromTheQuantumStructure} I showed that if such a method results in a unique TPS, that TPS can't have observable consequences.

In Section \sref{s:proof-counting} I show that no method to obtain a TPS only from the spectrum of the Hamiltonian can have a unique result.
In Section \sref{s:entanglement} I show that even if the resulting TPS were unique, such a Hamiltonian would be unable to entangle subsystems.
So even if ``unique'' would mean ``unique up to an equivalence'' in whatever sense, such a TPS won't be like the TPSs known from physics.
I discuss the implications for some approaches to emergent space or spacetime.

\section{Proof of non-uniqueness}
\label{s:proof-counting}

\begin{theorem}
\label{thm:nonuniqueness}
No method to obtain the TPS from the Hamiltonian's spectrum can give a unique result.
The number of additional continuous parameters needed to fix the TPS uniquely is exponential in the number of qudits.
\end{theorem}
\begin{proof}
Any TPS obtained uniquely only from the Hamiltonian's spectrum should be preserved by any unitary symmetry transformation $\oper{S}$ preserving the Hamiltonian, $\oper{S}\oper{H}\oper{S}^\dagger=\oper{H}$. Otherwise, this indicates that the method to obtain the TPS is based on additional hidden assumptions, introducing more information than the Hamiltonian's spectrum.
To see this, consider that the method to obtain the TPS consists of calculations done in a basis of $\hilbert$ in which $\oper{H}$ is diagonal, with the eigenvalues sorted in ascending order.
If the resulting TPS depends only on the spectrum of $\oper{H}$, we should obtain the same TPS if we use a different basis in which $\oper{H}$ has the same form.
Any such basis can be obtained from the former basis by a unitary transformation that commutes with $\oper{H}$.

The TPS can be specified uniquely in terms of algebras of observables $\mc{A}_1,\mc{A}_2,\ldots$, \ie algebras of operators on $\hilbert$, required to generate the full algebra of observables of $\hilbert$, and to satisfy, for all $j\neq k$, the conditions \cite{ZanardiLidarLloyd2004QuantumTensorProductStructuresAreObservableInduced}
\begin{equation}
\label{eq:TPS-alg-obs}
\left[\mc{A}_j,\mc{A}_k\right]=\oper{0} \tn{ and } \mc{A}_j\cap\mc{A}_k=\C\oper{1}_{\hilbert}.
\end{equation}

Therefore, the TPS resulting by transforming the original TPS using a unitary transformation $\oper{S}$ is specified by the algebra of observables
\begin{equation}
\label{eq:TPS-transformed}
\left\{\mc{A}'_1,\mc{A}'_2,\ldots\right\}=\left\{\oper{S}\mc{A}_1\oper{S}^\dagger,\oper{S}\mc{A}_2\oper{S}^\dagger,\ldots\right\}.
\end{equation}

If the TPS were unique, this would imply that, for any $\oper{S}$ satisfying $\oper{S}\oper{H}\oper{S}^\dagger=\oper{H}$, the two sets of algebras of observables are identical,
\begin{equation}
\label{eq:TPS-invar}
\left\{\mc{A}_1,\mc{A}_2,\ldots\right\}=\left\{\oper{S}\mc{A}_1\oper{S}^\dagger,\oper{S}\mc{A}_2\oper{S}^\dagger,\ldots\right\}.
\end{equation}

Therefore, a necessary condition for the uniqueness of the TPS is that the symmetries preserving the Hamiltonian should also preserve the TPS. That is,
\begin{equation}
\label{eq:uniq-nec-cond}
\U\(\oper{H}\)\subseteq\U_{\tn{TPS}},
\end{equation}
where $\U\(\oper{H}\)$ denotes the group of symmetry transformations preserving the Hamiltonian,
and $\U_{\tn{TPS}}$ denotes the group of unitary transformations preserving the TPS.

To see if the necessary condition \eqref{eq:uniq-nec-cond} can be satisfied, we can compare the dimensions of these two groups or of some of their subgroups.

The TPS is invariant under unitary transformations of $\hilbert$ acting restrictively on each of the factor Hilbert spaces $\hilbert_j$, and under the permutations of the factors $\hilbert_j$ of the same dimension \cite{ZanardiLidarLloyd2004QuantumTensorProductStructuresAreObservableInduced,CotlerEtAl2019LocalityFromSpectrum}.
The group of local unitary transformations acting only on $\hilbert_j$ is $\U\(\hilbert_j\)$.
But since $\U\(\hilbert_j\)\cap\U\(\hilbert_k\)=\U(1)$ for $j\neq k$, local phase changes are equivalent to the global ones.
Taking this (and the fact that group of permutations of the factors doesn't affect the dimension) into account, it follows that
\begin{equation}
\label{eq:group-TPS}
\begin{aligned}
\dim\U_{\tn{TPS}}
&\leq\underbrace{\dim\U(1)}_{\tn{global}}+\sum_{j=1}^n\(\dim\U\(\hilbert_j\)-\underbrace{\dim\U(1)}_{\tn{local}}\) \\
&=\sum_{j=1}^n d_j^2-n+1,
\end{aligned}
\end{equation}
where $d_j:=\dim\hilbert_j$, for all $j$.

Since the dimension of the Hilbert space $\hilbert$ is $\dim\hilbert=d_1\cdot\ldots\cdot d_n$, the dimension of the group of symmetry transformations preserving the Hamiltonian satisfies
\begin{equation}
\label{eq:dim-group-H}
d_1\cdot\ldots\cdot d_n\leq\dim\U\(\oper{H}\)\leq d_1^2\cdot\ldots\cdot d_n^2,
\end{equation}
The lower bound is attained if all eigenvalues of $\oper{H}$ are distinct, and the upper bound is attained if all eigenvalues are equal.

If the necessary condition \eqref{eq:uniq-nec-cond} were satisfied, any commutative subgroup of $\U\(\oper{H}\)$ would be included in a commutative subgroup of $\U_{\tn{TPS}}$.
I will show that this doesn't happen.

Since $\dim\hilbert=d_1\cdot\ldots\cdot d_n$, the dimension of any commutative subgroup of $\U\(\oper{H}\)$ is at most
\begin{equation}
\label{eq:dim-max-commutative-H}
\tn{D}_{\oper{H}}=d_1\cdot\ldots\cdot d_n.
\end{equation}

But the dimension of any commutative subgroup of $\U_{\tn{TPS}}$ is at most
\begin{equation}
\label{eq:dim-max-commutative-TPS}
\tn{D}_{\tn{TPS}}
=\sum_{j=1}^n (d_j-1)+1
=\sum_{j=1}^n d_j-n+1.
\end{equation}

Since $d_j\geq2$ for all $j$, for all $n>1$ we have
\begin{equation}
\label{eq:dim-max-commutative-compare}
d_1\cdot\ldots\cdot d_n>\sum_{j=1}^n d_j-n+1,
\end{equation}
so $\tn{D}_{\oper{H}}>\tn{D}_{\tn{TPS}}.
$
Then, the necessary condition \eqref{eq:uniq-nec-cond} can't be satisfied, and there are symmetries preserving the Hamiltonian that don't preserve the TPS.
Therefore, the TPS can't be uniquely determined by the Hamiltonian's spectrum.

To obtain a unique TPS, additional parameters need to be fixed, and they can't come from the spectrum.

Let's see how fast the number of additional parameters needed grows as $n$ grows.
If all subsystems are qubits or qudits of the same dimension $d_j=d$, the dimensions of the two groups $\U_{\tn{TPS}}$ and $\U\(\oper{H}\)$ are
\begin{equation}
\label{eq:dim-groups-d}
\begin{array}{lll}
\dim & \U_{\tn{TPS}} & =n(d^2-1)+1, \\
\dim & \U\(\oper{H}\) & \geq d^n. \\
\end{array}
\end{equation}

We see that $\dim\U_{\tn{TPS}}$ is linear in $n$, while $\dim\U\(\oper{H}\)$ is exponential in $n$.
Therefore, the number of additional continuous parameters needed to fix the TPS uniquely is exponential in the number of qudits.
\end{proof}

\section{Discussion: TPS, entanglement, and emergent spacetime}
\label{s:entanglement}

Now we'll see that, even if the math would have allowed some Hamiltonians to determine a unique TPS, such a Hamiltonian would not generate a unitary dynamics able to entangle subsystems.
Of course, we now know from Theorem \ref{thm:nonuniqueness} that the TPS can't be unique, but it's interesting to explore some physical reasons why this idea couldn't work to begin with.
This will also help see that, even if we would hope to get rid of non-uniqueness by taking equivalence classes of the resulting TPSs, this would not make sense physically.

The operators $e^{-\frac{i}{\hbar}\oper{H}t}$, $t\in\R$, commute with $\oper{H}$, so if the TPS were unique, when they are used as symmetry transformations, they should preserve it. But they also act as evolution operators of the system,
\begin{equation}
\label{eq:evol}
\oper{\rho}(t)=e^{-\frac{i}{\hbar}\oper{H}t}\oper{\rho}(0)e^{\frac{i}{\hbar}\oper{H}t}.
\end{equation}

Consider a state that is initially separable,
\begin{equation}
\label{eq:init-state}
\oper{\rho}(0)=\oper{\rho}_1(0)\otimes\oper{\rho}_2(0)\otimes\ldots.
\end{equation}

A state of the form $\ldots\otimes\oper{1}\otimes\oper{\rho}_j(0)\otimes\oper{1}\otimes\ldots$ is from the algebra of observables $\mc{A}_j$.
From equation \eqref{eq:TPS-invar} it follows that it evolves into a state of the form
\begin{equation}
\label{eq:evol-state-j}
\ldots\otimes\oper{1}\otimes\oper{\rho}'_{\sigma(j)}(t) \otimes\oper{1}\otimes\ldots.
\end{equation}

Since $\oper{\rho}(t)$ is continuous in $t$, $\oper{\rho}'_{\sigma(j)}(t)$ must retain its position $j$, so $\sigma(j)=j$ for all $j$.
From equation \eqref{eq:evol-state-j} and the first condition in \eqref{eq:TPS-alg-obs},
\begin{equation}
\label{eq:evol-state}
\oper{\rho}(t)=\oper{\rho}'_1(t)\otimes\oper{\rho}'_2(t)\otimes\ldots.
\end{equation}

Therefore, such a unitary evolution can't entangle the subsystems determined by the unique TPS.
But in the real world, unitary evolution can entangle subsystems.
Therefore, regardless of the method we use to get a unique TPS only from the spectrum, such a TPS can't be the TPS of systems that can change their entanglement.

In addition, the entropy of entanglement is invariant under the symmetry transformations from $\U_{\tn{TPS}}$.
Therefore, two TPSs determining different entanglement entropies should be nonequivalent.
Since for unitary evolution the entropy of entanglement is continuous in time, symmetry transformations of the form $e^{-\frac{i}{\hbar}\oper{H}t}$ that change the entropy of entanglement can't be from the group $\U_{\tn{TPS}}$. They transform the TPS into infinitely many nonequivalent TPSs resulting from the spectrum by the same method.

I chose the operators $e^{-\frac{i}{\hbar}\oper{H}t}$ as an example, a useful one since we know that the entanglement can be changed by unitary time evolution. But there are infinitely many other one-parameter groups that commute with $\oper{H}$, each of them being able to generate TPSs with different entanglement invariants, and therefore they can't be equivalent.

One may think to avoid the existence of the numerous non-equivalent TPSs implied by the fact that entanglement properties can change, by declaring all of them equivalent.
That is, declare that all TPSs obtained from one another by the symmetry transformations $e^{-\frac{i}{\hbar}\oper{H}t}$ are physically the same TPS.
But the entropy of entanglement is relative to the TPS, and any such equivalence class of TPSs would contain TPSs with respect to which the entanglement properties are different. So any TPS constructed by such an equivalence would be unable to exhibit changes in the entropy of entanglement, because all these changes would stay within the same equivalence class.
Different states with different entanglement properties would be declared to be physically equivalent.
And we know that in reality the entropy of entanglement can be changed by unitary evolution, and different states can have different entanglement properties.
So whatever we get from a construction of the TPS from the spectrum and an equivalence relation, it would not be the TPS that we know from physics.

It was proposed that, if a TPS emerges uniquely from the spectrum, the factors of the TPS can be identified with the points of the space, or whatever discrete building blocks of space we think should replace points in a quantum theory of gravity
\cite{CarrollSingh2019MadDogEverettianism,Carroll2021RealityAsAVectorInHilbertSpace}. 
The idea to recover the continuous $3$d-space from discrete building blocks is clever and it makes sense \cite{Finkelstein1969SpaceTimeCode}, and indeed these building blocks correspond to tensor product structures.

Then, if the TPS were unique, the distances between the resulting points of space could be recovered by using information that is not contained in the spectrum of the Hamiltonian, but that can be extracted from the state $\oper{\rho}$ and its relation with the TPS. Then maybe the $3$d space manifold structure would emerge at large scales. Carroll and Singh use the \emph{mutual information} for this \cite{CarrollSingh2019MadDogEverettianism,Carroll2021RealityAsAVectorInHilbertSpace}.

But now we've seen that even if the resulting TPS were unique, unitary evolution wouldn't be able to change the entanglement invariants of the TPS. The resulting spacetime metric, defined by using the mutual information, would be degenerate in the time direction, so the resulting geometry would be \emph{Carrollian} (``Carrollian spaces'' are named after Lewis Carroll \cite{LevyLeblond1965NouvelleLimiteNonRelativisteGroupePoincare,DuvalGibbonsHorvathyZhang2014CarrollVersusNewtonAndGalilei}). Everything would be frozen in time, and space would be absolute.

But more importantly, if the building blocks of space correspond to the factors of the TPS, Theorem \ref{thm:nonuniqueness} implies that they can't emerge uniquely from the spectrum.
For the infinite-dimensional Hilbert space case, infinitely many examples of systems with the same Hamiltonian and unit vector, but with different TPSs and different dimensions of space are already known to exist \cite{Stoica2022VersatilityOfTranslations}.
Moreover, even if the TPS is given along with the Hamiltonian and a unit vector, this is insufficient to describe our world \cite{Stoica2023PrinceAndPauperQuantumParadoxHilbertSpaceFundamentalism}.
Additional observables (along with their physical meaning \cite{Stoica2023AreObserversReducibleToStructures}) are needed to begin with.


\end{document}